%% file: master.tex
\documentclass{llncs}
\usepackage[latin1]{inputenc} 
\usepackage{amssymb}
\usepackage{amsmath}
\usepackage{epsfig}
\usepackage{color}
\usepackage[textwidth=\oddsidemargin+6ex,textsize=footnotesize]{todonotes}

\newcommand{\rpst}{{\bf MSRPP}}

\usepackage{caption}                                                                                                                                                                     
                                                                                                                                                                                         
\definecolor{darkorange}{rgb}{1.0,0.5,0}
\definecolor{darkgreen}{rgb}{0,0.5,0}
\newcommand{\edo}{\end{document}} % to interrupt compilation prematurely
\begin{document}

\title{Counterexample for the 2-approximation of finding partitions of rectilinear polygons with minimum stabbing number}

\author{Breno Piva\inst{1} \and Cid C. de Souza\inst{2}}

\institute{Universidade Federal de Sergipe, Departamento de Computa\c{c}\~{a}o, Av. Marechal Rondon, s/n Jardim Rosa Elze, 49100-000, S\~{a}o Crist\'{o}v\~{a}o, Sergipe, Brasil. \email{bpiva@ufs.br} \and Universidade de Campinas, Instituto de Computa\c{c}\~{a}o, Campinas, S\~{a}o Paulo, Brasil. \email{cid@ic.unicamp.br}}

\maketitle

\begin{abstract}
This paper presents a counterexample for the approximation algorithm proposed by Durocher and Mehrabi \cite{DurocherM12} for the general problem of finding a rectangular partition of a rectilinear polygon with minimum stabbing number.
\end{abstract}

\input{intro}

\input{ipmodels}

\input{thecounterexample}

\input{concl}

\bibliographystyle{plain}
\bibliography{biblio}

\input{appendix}

\end{document}

%% file: intro.tex
\section{Introduction}
\label{sec:intro}

Given a  rectilinear polygon  $P$ and a  rectangular partition  $R$ of
$P$, a segment  is said to be \textbf{rectilinear} relative  to $P$ if
it  is  parallel  to  one  of  $P$'s sides.   Let  $s$  be  a  maximal
rectilinear  line segment  inside  $P$.  The  stabbing  number of  $s$
relative to $R$ is defined as the number of rectangles of $R$ that $s$
intersects.  The stabbing number of $R$ is the largest stabbing number
of  a maximal  rectilinear line  segment  inside $P$.  
The Minimum Stabbing Rectangular Partition Problem (\rpst) consists in
finding  a  rectangular  partition  $R$ of  $P$  having  the  smallest
possible stabbing number.  Figure~\ref{fig:random-20-17-0} illustrates
these definitions.

Variants of the problem arise  from restricting the set of rectangular
partitions that are  considered to be valid. One of  these variants is
called the {\em conforming case}, in  which every edge in the solution
must be maximal, i.e., both of  its endpoints must touch the border of
the polygon.  
%%%%%%%%%%%%%%%%%%%%%%%%%%%%%%%%%%%%%%%%%%%%%%%%%%%%%%%%%%%%%%%%%%%%%%
% CID: not good yet ...
%%%%%%%%%%%%%%%%%%%%%%%%%%%%%%%%%%%%%%%%%%%%%%%%%%%%%%%%%%%%%%%%%%%%%
For this  problem, in \cite{DurocherM12},  Durocher et al.  propose an
integer  programming model  for the  conforming case  where there  are
exactly two  edges (that can  be in  the solution) having  each reflex
vertex  as endpoint.  Thus,  there are  also  precisely two  variables
associated to each reflex vertex.

\begin{figure}[h!]
\centering \epsfig{figure=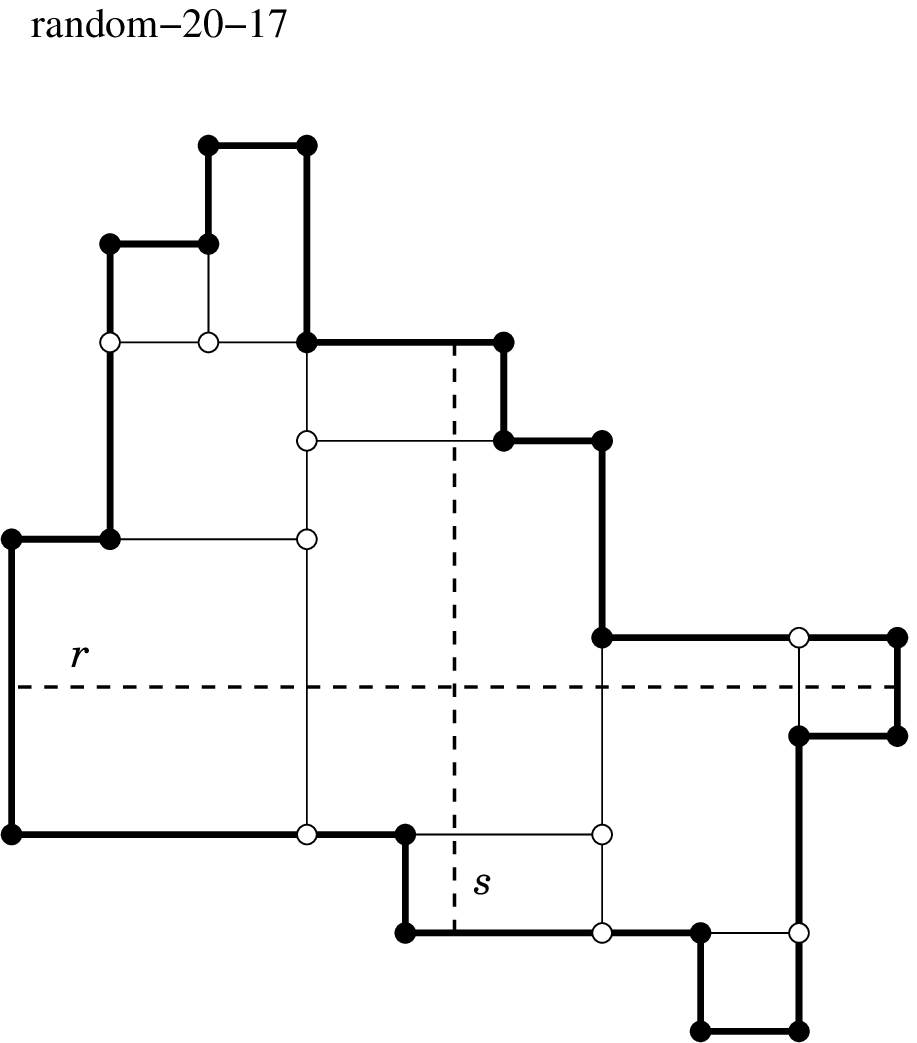,width=0.5\textwidth}
\caption{A  rectilinear  polygon  with   a  rectangular  partition  of
  stabbing number $4$. The  dashed lines represent maximal rectilinear
  line segments inside  the polygon.  Segment $r$  has stabbing number
  $4$ while segment $s$ has stabbing number $3$.}
\label{fig:random-20-17-0}
\end{figure}

In \cite{DurocherM12} a $2$-approximation algorithm is  presented for the conforming case of partitions of rectilinear polygons with minimum
stabbing  number. That approximation algorithm is based in a rounding of the variables. In  the \emph{Conclusion}  section  of the  article, it  is stated that the algorithm could be extended for the general case using
a formulation described informally and the same rounding rules used in
the conforming case.

In  this   paper  we   show  that  the   algorithm  as   described  in
\cite{DurocherM12}  cannot give  a $2$-approximation  for the  general
case of the (\rpst). This is done  by means of a counterexample to the
referred algorithm.

%% file: ipmodels.tex
\section{IP Models}
\label{sec:ipmodels}

The \rpst\  can be  modelled via  integer programming  in a  number of
different ways.   In this  section we  present two  such models for the general case of \rpst\ in an
attempt to formalize the description given in \cite{DurocherM12}.  But
first, we need some definitions.

Let  $P$ be  a  rectilinear polygon,  input of  the  \rpst. Define  as
$V_r^P$ the set of \textbf{reflex vertices} of $P$, i.e., those having
internal  angles equal  to $3  \pi /  2$. Let  $V_c^P$ be  the set  of
vertices of $P$  that are not reflex. Denote by  $grid(P)$, the set of
all maximal rectilinear line segments in  the interior of $P$ having a
vertex in $V_r^P$ as  one of its endpoints. Let $V_s^P$  be the set of
points in the intersection of two  segments in $grid(P)$.  We refer to
these points as \textbf{Steiner Vertices}.  The points that are not in
$V_r^P$  or $V_c^P$  and  are  in the  intersection  of  a segment  in
$grid(P)$ and the  border of $P$ compose the set  $V_b^P$.  
Denote by $V^P$ the set resulting from the union of all the point sets
defined before, i.e., $V^P = V_r^P \cup V_c^P \cup V_s^P \cup V_b^P$.

Define $E_h^P$ as the set of line segments in the border of $P$ having
only two points in $V^P$ which are its extremities.
Any fragment of a segment in $grid(P)$ containing exactly two vertices
in $V^P$ is called an \textbf{internal  edge}. The set of all internal
edges is $E_i^P$ and the set of all  edges in $P$ is $E^P = E_h^P \cup
E_i^P$. A subset $E'^P$ of $E^P$  defines a \emph{knee} in a vertex $u
\in V_s^P \cup  V_r^P$ if exactly two  edges in $E'^P$ have  $u$ as an
endpoint and these  edges are orthogonal. A subset $E'^P$  of $E^P$ is
said to define an \emph{island} in a  vertex $u \in V_r^P$ if only one
edge of $E'^P$ have $u$ as an  endpoint. At last, if $ua$ and $ub$ are
two edges in  $E^P$ having a common endpoint $u$,  we denote the angle
between $ua$ and $ub$ by $\theta(ua,ub)$.

Now, we can formalize the model described in \cite{DurocherM12} as follows:

{\small
\begin{align}
(RPST)            && z ~=~ & \min k & \label{eq:rpstobj} \\
\mbox{subject to} %&& \nonumber \\
									&&			 & x_{ua} + x_{ub} \geq  1, & \forall\ u \in V_r^P \wedge ua, ub \in E_i^P, \label{eq:minmaxreflex} \\
                  &&       & x_{ua} + x_{ub} - x_{uc} \geq 0, & \forall\ u \in V_s^P, \forall\ ua, ub, uc \in E_i^P \nonumber \\
									&&   &&  \mbox{with}~\theta(ua,ub) = \pi/2 \label{eq:minmaxsteiner}, \\
									&&			 & \sum_{\substack{uv \in E_i^P \\ uv \bigcap s \neq \emptyset}} x_{uv}  \leq k-1, & \forall\ {s\in L}, \label{eq:rpststab} \\
									&&       & x_{uv} \in \mathbb{B} & \forall\ uv \in E_i^P, \label{eq:rpstintegral}\\
                  &&       & k \in \mathbb{Z}. \label{eq:stabintegral}
\end{align}
}

In the  model above,  we have  one binary  variable $x_{uv}$  for each
internal edge  $uv$ in  $P$ which  is set to  $1$ if  and only  if the
corresponding   edge  is   in  the   rectangular  partition   of  $P$.
Constraints~\eqref{eq:minmaxreflex} ensure that  the solution does not
contain       a       knee       in       a       reflex       vertex.
Inequalities~\eqref{eq:minmaxsteiner}  impose that  the solution  does
not   form   a   knee   or    an   island   in   a   Steiner   vertex.
Inequalities~\eqref{eq:rpststab}   relate  the   $x$  variables   with
variable $k$, which represents the stabbing number of the solution. As
a  consequence,   the  objective  function~\eqref{eq:rpstobj}   is  to
minimize      $k$.       Finally,     \eqref{eq:rpstintegral}      and
\eqref{eq:stabintegral}   are   integrality   restrictions   for   the
variables.   Figure~\ref{fig:random_20_17_1} shows  an instance  of the
\rpst\ (called  {\tt random-20-17})  with 62  internal edges  and their
corresponding variables.

\begin{figure}[h!]
	\centering \epsfig{figure=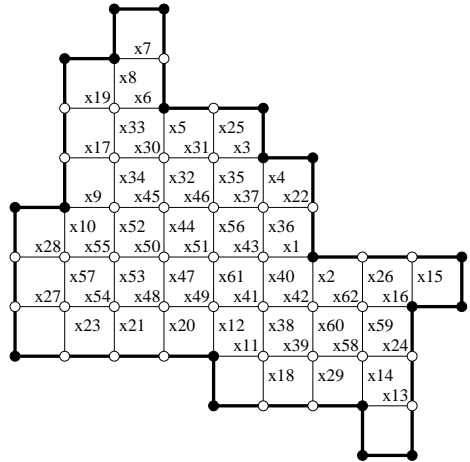,width=0.5\textwidth}
	\caption{Instance {\tt random-20-17} with 62 internal edges and the corresponding variables.}
\label{fig:random_20_17_1}
\end{figure}

As stated before,  the $(RPST)$ model above is not  the only model for
the problem and next we show another way of modelling it. However, to guarantee the correctness of the model we must first prove a property of optimal solutions for the \rpst. The following proposition is a generalization of \emph{Observation 1} in \cite{DurocherM12}.

\begin{proposition}
\label{prop:obs1}
Any rectilinear polygon $P$ has an optimal rectangular partition $R$ in which every maximal segment of $R$ has at least one reflex vertex of $P$ as an endpoint.
\end{proposition}
\begin{proof}
Let $R$ be a rectangular partition of a rectilinear polygon $P$. Let $e$ be a maximal segment in $R$ having $a$ and $b$ as its endpoints. Suppose neither $a$ nor $b$ are reflex vertices. Since $e$ is maximal and $R$ is a rectangular partition, both endpoints of $e$ must lie in segments perpendicular to $e$.

Now, since $R$ is a rectangular partition, $e$ define two minimal rectangles (each one possibly containing other rectangles) having $e$ as one of its sides, let us denote them by $r_1$ and $r_2$. There are three cases to consider.

The first case consists of $r_1$ and $r_2$ been empty rectangles, i.e., neither $r_1$ nor $r_2$ contain other rectangles. Therefore, the removal of $e$ unite these rectangles, composing a single rectangle. Therefore, $R\setminus e$ is still a rectangular partition. It is clear that removing a segment cannot increase the stabbing number of the solution. Thus, if $R$ is an optimal solution, so is $R\setminus e$.

The second case to consider is when only one of $r_1$ or $r_2$ contains other rectangles. Suppose without loss of generality that $r_1$ is the one containing other rectangles. Now, we can drag $e$ towards $r_1$, shrinking any segment with an endpoint in $e$, until $e$ meets a reflex vertex or the border of $P$. In the latter case, $e$ is merged to the border of $P$. It is easy to see that the result of this dragging operation is also a rectangular partition besides, the only stabbing segments affected by this operation are the ones parallel to $e$ and their stabbing number cannot increase. Therefore, as $R$ is optimal, so must be the new solution.

At last, we must consider the case where both $r_1$ and $r_2$ contain other rectangles. Suppose without loss of generality that the number of segments in $r_1$ having an endpoint in $e$ (thus, perpendicular to it) is greater or equal than the number of segments with these characteristics in $r_2$. Then, again, we can drag $e$ towards $r_1$, shrinking any segment with an endpoint in $e$, until $e$ meets either a segment parallel to $e$ or a reflex vertex or the border of $P$.
If a parallel segment is met, $e$ is merged to it and the process is repeated until a reflex vertex or the border of $P$ is met.
In case the border of $P$ is met, $e$ ceases to exist together with the segments in the space between $e$ and the border.
Once again, the dragging operation results in a rectangular partition of $P$ and the only stabbing segments affected by this operation are parallel to $e$. But, as the number of segments in $r_1$ is greater or equal than the number of segments in $r_2$, one can see that the stabbing number of the new rectangular partition cannot be greater than that of $R$.

Ergo, there is always an optimal rectangular partition where every maximal segment has at least one reflex vertex of $P$ as an endpoint.\qed
\end{proof}

In the next  model, given the same definitions as  before, we consider
the set $E_e^P$  of rectilinear segments $uv$ where $u  \in V_r^P$ and
$v \in  V^P$. Notice  that a  segment of $E_e^P$  can be  comprised of
several consecutive segments  of $E_i^P$.  Hence, we  call $E_e^P$ the
\textbf{extended  edge} set.   In  the formulation  below,  we have  a
variable $x_{uv}$ for each edge in $E_e^P$ and from Proposition~\ref{prop:obs1} it is easy to notice that this set of variables is sufficient to provide optimal rectangular partitions.

{\small
\begin{align}
(RPST2)            && z ~=~ & \min k & \label{eq:rpstobj2} \\
\mbox{subject to}  && \nonumber \\
									 &&			  & \sum_{ua \in E_e^P} x_{ua} \geq  1, && \forall\ u \in V_r^P \label{eq:reflex2} \\
                   &&       & x_{ab} + x_{uv} \leq 1, && \forall\ ab, uv : ab \cap uv \neq \emptyset\ \wedge \nonumber \\
									 &&       & && \wedge ab \cap uv \neq a, b, u\ or\ v \label{eq:planar2} \\
									 &&			
%& \sum_{\tiny{
& \sum_{
\substack{\theta(uv, ab) = \pi/2~\wedge\\  \wedge~b \in uv~\wedge~b \neq u~\wedge~b \neq v}
}
%\begin{array}{c}\theta(uv, ab) = \pi/2 \wedge\ b \in uv\ \wedge\\ \wedge\ b \neq u \wedge\ b \neq v\end{array}}} 
x_{uv} - x_{ab} \geq 0, && \forall\ a \in V_r^P, b \in V_s^P \label{eq:steiner2} \\ 
  								 &&			 & \sum_{uv \in E_e^P : uv \bigcap s \neq \emptyset} x_{uv} \leq k-1,  && \forall\ {s\in L} \label{eq:rpststab2} \\
									&&       & x_{uv} \in \mathbb{B} && \forall\ uv \in E_e^P. \label{eq:rpstintegral2}\\
                  &&       & k \in \mathbb{Z} &&& & && \label{eq:stabintegral2}
\end{align}
}

In  this model,  inequalities  \eqref{eq:reflex2}  guarantee that  the
solution  does not  contain a  knee in  a reflex  vertex.  Constraints
\eqref{eq:planar2} enforce  planarity (two  segments of  the partition
can       only        intersect       at        their       extremes).
Constraints~\eqref{eq:steiner2}  prevent the  existence  of knees  and
islands in  a Steiner  vertex.  Finally, \eqref{eq:rpststab2}  are the
stabbing      constraints      and     \eqref{eq:rpstintegral2}   and
\eqref{eq:stabintegral2}       are      integrality       constraints.
Figure~\ref{fig:random_20_17_2} shows  instance {\tt  random-20-17} with
42 internal edges and the corresponding variables.

\begin{figure}[htp]
	\centering \epsfig{figure=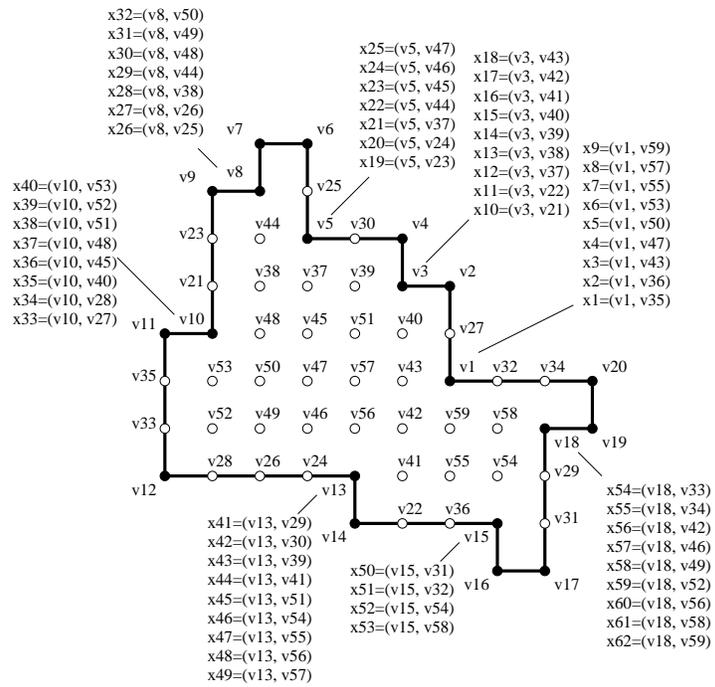,width=0.8\textwidth}
	\caption{Instance {\tt random-20-17} with its extended edges and corresponding variables.}
\label{fig:random_20_17_2}
\end{figure}

%% file: thecounterexample.tex
\section{The Counterexample}
\label{sec:thecounterexample}

Before discussing  the counterexample,  we first present  the rounding
scheme proposed  in \cite{DurocherM12} for the  conforming case. Once the optimum of the linear relaxation is computed, the rules
for rounding  variables in  the conforming case  are really  simple: a
variable corresponding to a horizontal segment is rounded down to zero
if its value  is smaller than or  equal to $0.5$ and is  rounded up to
one if its value is greater than $0.5$.  A variable corresponding to a
vertical segment is rounded down to  zero if its value is smaller than
$0.5$ and is rounded  up to one if its value is  greater than or equal
to $0.5$.

In  the \emph{Conclusion}  section  of \cite{DurocherM12},  a  model for  the
general  (non-conforming)  case  is described  informally.   From  the
discussion,  apparently  such  model  is equivalent  to  the  $(RPST)$
formulation  given in  Section~\ref{sec:ipmodels}.   According to  the
authors, the same rounding rules used in the conforming case provide a
2-approximation for the general case.

The rounding  rules do  not mention  what should  be done  for Steiner
vertices, and  no guarantee  is given that  applying them  directly in
these situations will avoid the formation of a knee or an island.
In  fact, the  instance displayed  in Figure~\ref{fig:counterexample-20-17-1}
shows that this cannot always be done without sacrificing feasibility.
In this figure,  the optimal values of the  variables corresponding to
edges    incident     to    Steiner    vertex    {\tt     v39}    (see
Figure~\ref{fig:random_20_17_2})  after solving  the linear  relaxation
associated  to instance  {\tt  random-20-17} are  given.   As only  the
variable corresponding  to one vertical  edge incident to  that vertex
has value equal  to $0.5$ and the other three  are smaller than $0.5$,
rounding  according  to  that  rule  would  result  in  an  island  at
$v37$. Therefore, the  set of edges obtaining after  rounding does not
form a rectangular partition.

\begin{figure}[htp]
	\centering \epsfig{figure=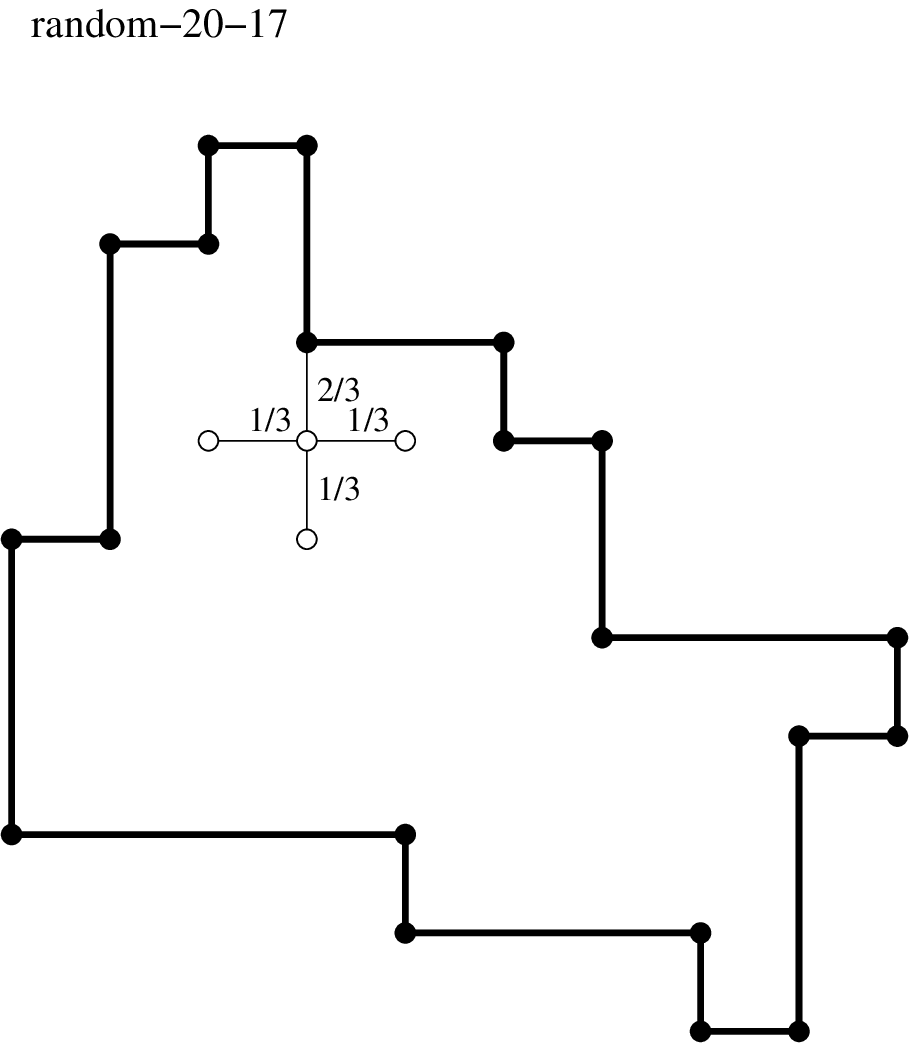,width=0.5\textwidth}
	\caption{Values of variables corresponding to edges incident to a Steiner vertex after solving linear relaxation.}
\label{fig:counterexample-20-17-1}
\end{figure}

It is  however possible that  we misinterpreted the model  the authors
were thinking of (although there is evidence in contrary) and the idea
is actually  to define variables  corresponding to all edges  having a
reflex vertex  as one of its  endpoints. If so, the  formulation would
look like $(RPST2)$ model in the previous section. In this alternative
formulation, rounding the variables using that rule does not cause the
same  problem as  before since  every variable  correspond to  an edge
having a reflex vertex as endpoint.

Contrary to what  happens in the conforming case,  however, the reflex
vertices  here have  more than  two incident  edges. Therefore,  it is
possible that the  solution of the linear relaxation  result in values
smaller than  $0.5$ for all  the variables corresponding to  the edges
incident  to a  certain  reflex  vertex. Thus,  the  rounding of  such
solution would result in a partition having a knee in a reflex vertex.

The situation  described above occurs  in practice with  instance {\tt
  random-20-17}, as shown in Figure~\ref{fig:counterexample-20-17-2}.
Consider the  edges incident  to vertex {\tt  v5}. All  the associated
variables incident to this vertex have value smaller than $0.5$.
As  consequence,  they will  be  rounded  to  zero, resulting  in  the
formation  of a  knee at  {\tt v5}  and, therefore,  in an  infeasible
solution.

\begin{figure}[htp]
	\centering \epsfig{figure=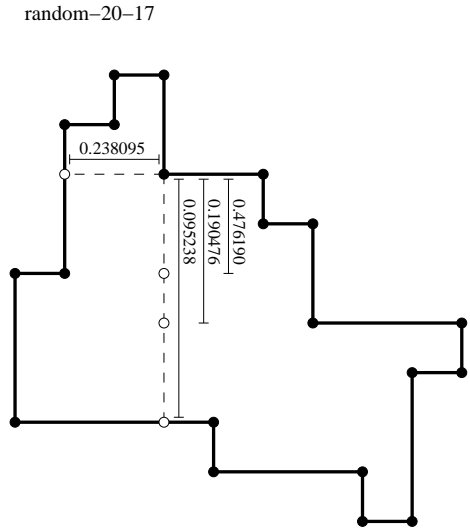,width=0.5\textwidth}
	\caption{Values of variables corresponding to edges incident to a reflex vertex after solving linear relaxation.}
\label{fig:counterexample-20-17-2}
\end{figure}

%% file: concl.tex
\section{Conclusion}
\label{sec:concl}

From  the counterexample presented in Section~\ref{sec:thecounterexample},  we conclude  that  it remains  open whether a 2-approximation for the \rpst\ in the general case exists. It is, however, noteworthy that many other contributions are presented in \cite{DurocherM12} and none of them are diminished by this counterexample.

%% file: appendix.tex
\section*{Appendix}
\label{sec:appendix}

\appendix

\begin{verbatim}
File name: random-20-17.rect
Model: RPST
Vertex number: 59
Edge number: 62

Reading Problem stab
Problem Statistics
         231 (      0 spare) rows
          63 (      0 spare) structural columns
         752 (      0 spare) non-zero elements
Global Statistics
          63 entities        0 sets        0 set members
Minimizing MILP stab
Original problem has:
       231 rows           63 cols          752 elements        63 globals 

   Its         Obj Value      S   Ninf  Nneg        Sum Inf  Time
     0           .000000      D      1     0      24.000000     0
    74          2.416667      D      0     0        .000000     0
Optimal solution found
 *** Search unfinished ***    Time: 0
Number of integer feasible solutions found is 0
Best bound is     2.416667

Solution:

X1  =  0.583333  X2  =  0.416667  X3  =  0.095238  X4  =  0.916667  X5  =  0.666667
X6  =  0.333333  X7  =  0.428571  X8  =  0.571429  X9  =  0.904762  X10 =  0.845238
X11 =  1.000000  X12 =  0.571429  X13 =  0.000000  X14 =  1.000000  X15 =  0.000000
X16 =  1.000000  X17 = -0.000000  X18 =  0.500000  X19 =  0.238095  X20 = -0.000000
X21 = -0.000000  X22 =  0.333333  X23 =  0.428571  X24 =  1.000000  X25 =  0.238095
X26 = -0.000000  X27 =  0.285714  X28 =  0.130952  X29 =  0.500000  X30 =  0.333333
X31 =  0.333333  X32 =  0.333333  X33 =  0.333333  X34 =  0.333333  X35 =  0.095238
X36 =  0.583333  X37 =  0.333333  X38 =  0.500000  X39 =  0.500000  X40 =  0.297619
X41 =  0.285714  X42 =  0.583333  X43 =  0.285714  X44 = -0.000000  X45 =  0.333333
X46 =  0.333333  X47 = -0.000000  X48 =  0.285714  X49 =  0.285714  X50 =  0.285714
X51 =  0.285714  X52 =  0.571429  X53 =  0.285714  X54 =  0.285714  X55 =  0.571429
X56 =  0.000000  X57 =  0.714286  X58 =  1.000000  X59 = -0.000000  X60 =  0.500000
X61 =  0.285714  X62 =  1.000000  X63 =  3.000000  

***********************************************************************************

File name: random-20-17.rect
Model: RPST2
Vertex number: 59
Edge number: 62
Reading Problem stab
Problem Statistics
         336 (      0 spare) rows
          63 (      0 spare) structural columns
         996 (      0 spare) non-zero elements
Global Statistics
          63 entities        0 sets        0 set members
Minimizing MILP stab
Original problem has:
       336 rows           63 cols          996 elements        63 globals 
Crash basis containing 13 structural columns created

   Its         Obj Value      S   Ninf  Nneg        Sum Inf  Time
     0           .000000      D      1     0      24.000000     0
    66          2.413793      D      0     0        .000000     0
Optimal solution found
 *** Search unfinished ***    Time: 0
Number of integer feasible solutions found is 0
Best bound is     2.413793

Solution:

X1  =  0.190476  X2  =  0.380952  X3  =  0.142857  X4  = -0.000000  X5  = -0.000000
X6  = -0.000000  X7  = -0.000000  X8  = -0.000000  X9  =  0.285714  X10 = -0.000000
X11 =  0.142857  X12 =  0.142857  X13 = -0.000000  X14 =  0.000000  X15 =  0.523810
X16 = -0.000000  X17 = -0.000000  X18 =  0.190476  X19 =  0.238095  X20 =  0.095238
X21 =  0.000000  X22 = -0.000000  X23 =  0.476190  X24 = -0.000000  X25 =  0.190476
X26 =  0.380952  X27 =  0.190476  X28 =  0.000000  X29 =  0.238095  X30 = -0.000000
X31 = -0.000000  X32 =  0.190476  X33 =  0.523810  X34 =  0.809524  X35 = -0.000000
X36 =  0.285714  X37 =  0.380952  X38 =  0.000000  X39 =  0.000000  X40 =  0.190476
X41 =  0.619048  X42 = -0.000000  X43 = -0.000000  X44 =  0.142857  X45 = -0.000000
X46 = -0.000000  X47 = -0.000000  X48 =  0.380952  X49 =  0.000000  X50 =  0.380952
X51 = -0.000000  X52 =  0.619048  X53 = -0.000000  X54 = -0.000000  X55 =  0.095238
X56 =  0.142857  X57 = -0.000000  X58 = -0.000000  X59 =  0.380952  X60 = -0.000000
X61 =  0.000000  X62 =  0.380952  X63 =  3.000000  
\end{verbatim}